\documentclass{article} 
\usepackage{arxiv}
\usepackage{amsthm}

\usepackage{graphicx}
\usepackage{natbib}        
\usepackage{hyperref}
\usepackage{tikz-cd}
\usepackage{subfig}

\usepackage{xcolor}  

\usepackage{amssymb}

\usepackage{amsmath}
\usepackage{mathdots}

\usepackage{mathtools}

\usepackage{tensor}
\usepackage{xifthen}

\usepackage{urwchancal}
\DeclareFontFamily{OT1}{pzc}{}
\DeclareFontShape{OT1}{pzc}{m}{it}{<-> s * [1.10] pzcmi7t}{}
\DeclareMathAlphabet{\mathpzc}{OT1}{pzc}{m}{it}

\usepackage{graphicx}
\usepackage{ifpdf}
\ifpdf
\usepackage{epstopdf}
\PrependGraphicsExtensions{.svg}
\fi

\newtheorem{theorem}{Theorem}[section]
\newtheorem{lemma}[theorem]{Lemma}

\newtheorem{definition}[theorem]{Definition}
\newtheorem{example}[theorem]{Example}
\newtheorem{remark}[theorem]{Remark}

\providecommand{\R}{\mathbb{R}}

\providecommand{\SO}{\mathbf{SO}}

\providecommand{\SE}{\mathbf{SE}}

\providecommand{\grpG}{\mathbf{G}}

\providecommand{\gothg}{\mathfrak{g}}

\providecommand{\gothX}{\mathfrak{X}} 

\providecommand{\so}{\mathfrak{so}}

\providecommand{\calL}{\mathcal{L}}
\providecommand{\calM}{\mathcal{M}}
\providecommand{\calN}{\mathcal{N}}

\providecommand{\calR}{\mathcal{R}}

\providecommand{\calT}{\mathcal{T}}

\providecommand{\vecL}{\mathbb{L}}

\providecommand{\vecV}{\mathbb{V}}
\providecommand{\vecW}{\mathbb{W}}

\providecommand{\tT}{\mathrm{T}} 

\providecommand{\Id}{I} 

\providecommand{\tL}{\mathrm{L}} 
\providecommand{\tR}{\mathrm{R}} 

\DeclareMathOperator{\tr}{tr}

\DeclareMathOperator{\Ad}{Ad}
\DeclareMathOperator{\ad}{ad}

\DeclareFontEncoding{LS1}{}{}
\DeclareFontSubstitution{LS1}{stix}{m}{n}
\DeclareSymbolFont{stixletters}{LS1}{stix}{m}{it}
\DeclareMathAccent{\cev}{\mathord}{stixletters}{"91}
\DeclareMathAccent{\vec}{\mathord}{stixletters}{"92}
\DeclareMathAccent{\vecev}{\mathord}{stixletters}{"95}

\providecommand{\td}{\mathrm{d}}
\providecommand{\tD}{\mathrm{D}}
\providecommand{\ddt}{\frac{\td}{\td t}}

\usepackage{accents}
\usepackage{mathtools}
\makeatletter
\providecommand{\scirc}{%
    \hbox{\fontfamily{\rmdefault}\fontsize{0.4\dimexpr(\f@size pt)}{0}\selectfont{\raisebox{-0.52ex}[0ex][-0.52ex]{$\circ$}}}}

\makeatother

\makeatletter
\providecommand{\ucirc}{%
    \hbox{\fontfamily{\rmdefault}\fontsize{0.4\dimexpr(\f@size pt)}{0}\selectfont{\raisebox{0.0ex}[0ex][-0.52ex]{$\circ$}}}}

\makeatother

\mathchardef\mhyphen="2D

\providecommand{\idx}[5][]{
\ifthenelse{\isempty{#1}}
{\tensor*[_{#4}^{#3}]{#2}{_{#5}}}
{\tensor*[_{#4}^{#3}]{#2}{^{#1}_{#5}}}
}

\providecommand{\etal}{\textit{et al.~}}

\usepackage{url}

\newcommand{\pp}[2]{ \frac{\partial #1}{\partial #2}}
\newcommand{\Gstar}{\grpG^\ltimes_{\gothg^\ast}}
\newcommand{\gstar}{\gothg^\ltimes_{\gothg^\ast}}

\newcommand{\inert}{\mathbb{I}}

\newcommand{\dd}[2]{ \frac{\delta #1}{\delta #2}}
\newcommand{\ddP}[1]{ \frac{\delta #1}{\delta P}}

\newcommand{\evalat}[2]{\left. #1 \right|_{#2}}
\newcommand{\evaldat}[2]{\left. \tD_{#1} \right|_{#2}}

\newcommand{\publicationdetails}
{Under review}
\newcommand{\publicationversion}
{Preprint}

\begin{document}

\title{Equivariant Tracking Control for Fully Actuated Mechanical Systems on Matrix Lie Groups}
\headertitle{Equivariant Tracking Control for Fully Actuated Mechanical Systems on Matrix Lie Groups}

\author{
\href{}{\includegraphics[scale=0.06]{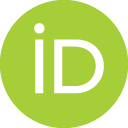}\hspace{1mm}
Matthew Hampsey}
\\
	Systems Theory and Robotics Group \\
	Australian National University \\
    ACT, 2601, Australia \\
	\texttt{matthew.hampsey@anu.edu.au} \\
	\And	\href{https://orcid.org/0000-0003-4391-7014}{\includegraphics[scale=0.06]{orcid.png}\hspace{1mm}
Pieter van Goor}
\\
	Robotics and Mechatronics Department,\\
    University of Twente,\\
   7500 AE Enschede, The Netherlands.\\
    \texttt{ p.c.h.vangoor@utwente.nl} \\
\And	\href{https://orcid.org/0000-0002-5746-7096}{\includegraphics[scale=0.06]{orcid.png}\hspace{1mm}
    Ravi Banavar}
\\
Centre for Systems and Control,\\
Indian Institute of Technology Bombay, India\\
\texttt{banavar@iitb.ac.in}
\And	\href{https://orcid.org/0000-0002-7803-2868}{\includegraphics[scale=0.06]{orcid.png}\hspace{1mm}
Robert Mahony}
\\
	Systems Theory and Robotics Group \\
	Australian National University \\
	ACT, 2601, Australia \\
	\texttt{robert.mahony@anu.edu.au} \\
}

\maketitle

\title{Equivariant Tracking Control for Fully Actuated Mechanical Systems on Matrix Lie Groups\thanksref{footnoteinfo}}

\begin{abstract}
Mechanical control systems such as aerial, marine, space, and terrestrial robots often naturally admit a state-space that has the structure of a Lie group.
The kinetic energy of such systems is commonly invariant to the induced action by the Lie group, and the system dynamics can be written as a coupled ordinary differential equation on the group and the dual space of its Lie algebra, termed a \emph{Lie-Poisson} system.
In this paper, we show that Lie-Poisson systems can also be written as a left-invariant system on a semi-direct Lie group structure placed on the trivialised cotangent bundle of the symmetry group.
The authors do not know of a prior reference for this observation and we are confident the insight has never been exploited in the context of tracking control.
We use this representation to build a right-invariant tracking error for the full state of a Lie-Poisson mechanical system and show that the error dynamics for this error are themselves of Lie-Poisson structure, albeit with time-varying inertia.
This allows us to tackle the general trajectory tracking problem using an energy shaping design metholodology.
To demonstrate the approach, we apply the proposed design methodology to a simple attitude tracking control.
\end{abstract}

\section{Introduction}
Trajectory tracking is a core system capability for all robotic systems \citep{SLOTINE1989509,spongRobotModelingControl2020}.
The classical solution to this problem uses a feed-forward computed torque control to enforce Euler-Lagrange or Hamiltonian error dynamics on an error system, and then applies passive damping injection to stabilise the resulting error system \citep{markiewiczAnalysisComputedTorque1973,paulRobotManipulatorsMathematics1981}.
The classical control design is effective, however, it depends heavily on the quality of the model in computing the feed-forward input, and this input cancels nonlinearities in the system to force the error system into the Euler-Lagrange structure, leading to potential wasted control action even if the model is exact.
To address the limitations of this design methodology, Fujimoto \etal~\citep{fujimotoTrajectoryTrackingControl2001} considered port Hamiltonian systems and used generalised canonical transforms to find error coordinates (by solving a specific set of partial differential equations)
for which the error system is also Hamiltonian.
The IDA-PBC algorithms of Ortega \etal~\citep{ORTEGA2004432,ortegaInterconnectionDampingAssignment2002} achieve a similar outcome for a class of Euler-Lagrange and port-Hamiltonian control systems.
Mahony \citep{mahonyNovelPassivitybasedTrajectory2019} introduced a 1-dimensional dynamics system that replicated the desired trajectory dynamics and used energy shaping and energy pumping between the true system and the new desired system to achieve tracking control.
All these papers are formulated in generalised coordinates while much of the nonlinearity of robot dynamics comes from their inherently non-linear state-spaces $\SE(3)$ or $\SO(3)$.
For systems with state-spaces on Lie groups, the group is normally a natural symmetry of the system, leading to the Lie-Poisson equations -- a left trivialised version of Hamilton's equations for a system on a Lie group with invariant inertia tensor \citep{marsden1994introduction}.
Attitude control on $\SO(3)$ for satellites and aerial robots
was studied in \citep{meyerDesignGlobalAnalysis1971a,Koditschek1989TheAO,chaturvediRigidBodyAttitudeControl2011,bayadiAlmostGlobalAttitude2014}.
General fully-actuated mechanical control systems on Lie groups were considered in \citep{bulloGeometricControlMechanical2005} and the stabilisation of general mechanical systems on Lie groups with nonholonomic constraints was studied in \citep{chandrasekaranGeometricPIDController2024}.
The geometry underlying these developments is well covered in \citep{murrayMathematicalIntroductionRobotic1994,lynchModernRobotics2017a}.
An interesting extension that is less commonly studied is that where the Hamiltonian is time-varying.
For example a satellite systems with changing configuration has time-varying inertia \citep{doi:10.2514/6.2014-4455}.
Another case study of interest is autonomous refueling \citep{haitaoModelingSimulationTimevarying2016} or control of a rocket \citep{gantmacherEquationsMotionRocket1950}, where the changing mass of the system leads to time-varying inertia.

Recently, the Lie group structure of robotic systems has been exploited in the development of equivariant observer and filter theory \citep{mahonyNonlinearComplementaryFilters2008,vangoorEquivariantFilterEqF2022,geEquivariantFilterDesign2022}.
Interestingly, the insight obtained by studying the observer problem has led to new perspectives on the tracking control problem as shown by recent work \citep{hampseyExploitingEquivarianceDesign2024,hampseySpatialGroupError2024,hampseyTrackingControlHomogeneous2023,weldeAlmostGlobalAsymptotic2024}.
In particular, the use of the spatial, or right-invariant, error contrasts to prior work in control literature, where the body, or left-invariant, error had been universally preferred.
For example, Bullo and Murray \citep{murray1995proportional} state that the ``reciprocal (right-invariant in our language) [$\ldots$] error [$\ldots$] depend[s] on the arbitrary choice of inertial frame'' to argue against using the error.
However, using the right invariant, or spatial, error in certain applications makes the error system synchronous \citep{trumpfAnalysisNonLinearAttitude2012,vangoorSynchronousModelsFundamental2025} and has been shown to lead to lower energy control action for the same performance \citep{hampseySpatialGroupError2024,hampseyExploitingDifferentSymmetries2023}.

In this paper, we derive a stabilising control law for time-varying Lie-Poisson systems.
We review the (trivialised) cotangent semi-direct symmetry \citep{marsdenSemidirectProductsReduction1984} and show that Lie-Poisson systems are left invariant with respect to this symmetry.
We then show that the induced right-invariant error system for the tracking task is also such a Lie-Poisson system, and thus the tracking task reduces to the stabilisation problem of a Lie-Poisson (error) system with time-varying inertia.
We use this to derive a tracking control law that consists of a feedforward term to compensate time-variation in the inertia, along with energy shaping and passive damping.
We explore the derived control with an example system and verify the performance with simulation studies.
The main results of the paper
\begin{enumerate}
    \item Introduce a symmetry group $\grpG^\ltimes_{\gothg^\ast}$ and group action for the trivialised cotangent bundle $\tT^\ast \grpG \cong \grpG \times \gothg^\ast$ that makes Lie-Poisson systems left-invariant.
    \item Introduce a right-invariant error between two dynamical Lie-Poisson systems and show that the resulting error system is itself a Lie-Poisson system with a time-varying inertia tensor.
    \item Propose a simple tracking control based on energy shaping and damping injection to stabilise the error system.
\end{enumerate}

The paper is structured as follows.
In Section \ref{sec:preliminaries}, the notation used throughout the paper is established, and some mathematical preliminaries required for the main development are revised.
In Section \ref{sec:lp-systems} we revise Lie-Poisson systems, the main dynamical system of interest in this paper.
We extend some known results regarding stabilisation to Lie-Poisson systems with time-varying inertias.
In Section \ref{sec:extended_lp_systems}, we develop the idea of extended Lie-Poisson systems in order to investigate the equivariant symmetries of these dynamical systems and define appropriate error functions.
In Section \ref{sec:tracking_control}, we show that the developed error system is itself a Lie-Poisson system, and apply the proposed stabilising control to this system to solve the tracking problem.
In Section \ref{sec:so3}, we present our proposed control design on a system posed on $\SO(3)$ and provide a tracking control simulation study.

\section{Preliminaries}
\label{sec:preliminaries}
Let $\vecV$ denote a (finite-dimensional) vector space and let  $\vecV^\ast$ denote its dual space; that is, the set of linear functionals acting on $\vecV$.
Let $P \in \vecV^\ast$ and $V \in \vecV$, then the contraction $P[V]$ of $P$ and $V$ is written $\langle P, V \rangle = \langle V, P \rangle$ (where $V$ and $(V^\ast)^\ast$ are identified).
If $\vecV, \vecW$ are vector spaces and $A : \vecV \to \vecW$ is a linear map, $W = A[V]$ where square brackets indicate linearity of the operator, then $A^\ast$ denotes the dual map $A^\ast : \vecW^\ast \to \vecV^\ast$ defined by
\begin{align}
    (A^\ast[P])[V] =  \langle A^\ast[P], V \rangle = \langle P, A[V] \rangle   = P[A[V]] \label{eq:dual_map}
\end{align}
for all $P \in \vecV^\ast$, $V \in \vecV$.
If compatible bases for $\vecV$ and $\vecV^\ast$ are used then $A^\ast = A^\top$, where we overload the notation $A$ to be the matrix representation of the operator as well as the abstract operator.

Let $\calM$ and $\calN$ denote smooth manifolds.
Denote the tangent space at an arbitrary point $Q \in \calM$ by $\tT_{Q}\calM$.
For a smooth function $h : \calM \to \calN$ the notation
\begin{align*}
    \evaldat{Z}{Q} h(Z):\tT_{Q}\calM\to \tT_{h(Q)}\calN \\
    V \mapsto \evaldat{Z}{Q} h(Z)[V],
\end{align*}
denotes the derivative of $h(Z)$ evaluated at $Z = Q$ in the direction $V \in \tT_{Q}\calM$.
When the basepoint and argument are implied the notation $\tD h$ will also be used for simplicity.

The space of smooth vector fields on $\calM$ is denoted by $\gothX(\calM)$.
Given a tangent bundle $\tT M$, there is an additional natural vector bundle (called the cotangent bundle) $\tT ^\ast \calM$ such that $\tT^\ast_p \calM$ is the dual space of $\tT_p \calM$.

Let $\grpG$ denote an $m$-dimensional matrix Lie group and denote the identity element with $\Id$.
The \emph{Lie algebra}, $\gothg$ of $\grpG$, is identified with the tangent space of $\grpG$ at identity, $\gothg\simeq \tT_{\Id}\grpG  \subset \R^{m \times m}$.
Given arbitrary $X \in \grpG$, \emph{left translation by $X$} is written
\mbox{$\mathrm \tL_X : \grpG \to \grpG$}, \mbox{$\mathrm \tL_X(Y) = XY.$}
This induces an isomorphism \mbox{$\tD \mathrm \tL_X : \gothg \to \tT_X \grpG$}.
For a matrix Lie group we will write $\tD \mathrm \tL_{X} U = XU$.
Similarly, \emph{right translation by $X$} \mbox{$\mathrm \tR_X(Y) = YX$}
induces an isomorphism \mbox{$\tD \mathrm \tR_X : \gothg \to \tT_X \grpG$},
written $\tD \mathrm \tR_{X}U = UX$.
Given $X \in \grpG$, the adjoint map \mbox{$\Ad_X: \gothg \to \gothg$} is defined by
$\Ad_X U := \tD \tL_X \, \tD \tR_{X^{-1}}[U]$ and can be written \mbox{$\Ad_X U = X U X^{-1}$} for a matrix Lie group.
In this paper, we use adjoint operators on different groups and the reader should take care to check which group and associated algebra the subscript and operand come from to infer the appropriate adjoint.
Given $U \in \gothg$, the ``little'' adjoint map \mbox{$\ad_U: \gothg \to \gothg$} is defined by the matrix commutator \mbox{$\ad_U V = UV - VU$}.
The operator $\ad^\ast : \gothg \times \gothg^\ast \to \gothg^\ast$ is defined by
\begin{align}
    \langle \ad^\ast_V P , U \rangle = \langle P, \ad_V U \rangle = \langle P, [V,U] \rangle.
    \label{eq:ad_star}
\end{align}
From \eqref{eq:dual_map}, if compatible bases are used for $\gothg$ and $\gothg^\ast$, then one has $\ad^\ast_V = \ad^\top_V$.
Given a function $h : \gothg^\ast \to \R$, one defines $\dd{h}{P} \in \gothg$ via the implicit relationship \citep{marsden1994introduction}
\begin{align}
    \left\langle W, \dd{h}{P} \right\rangle = \evalat{\ddt}{t=0} h(P + Wt)  = \tD h [W],
\label{eq:algebra_ident}
\end{align}
for all $W \in \gothg^\ast$.
Note that this is the same as identifying $\gothg \cong \tT^\ast_P (\gothg^\ast)$ for the derivative $\dd{h}{P} \coloneqq \tD h \in \tT^\ast_P (\gothg^\ast)$.

Let $\grpG$ be a Lie group and $\calM$ a smooth manifold.
A right action is a smooth function $\phi: \grpG \times \calM \to \calM$ satisfying the \emph{identity} and \emph{compatibility } properties:
\begin{align*}
     & \phi(\Id, Q)        = Q, & \phi(Y, \phi(X, Q)) = \phi(XY, Q)
\end{align*}
for all $Q \in \calM$ and $X \in \grpG$.
Given a fixed $Y \in \grpG$, the partial map $\phi_Y : \calM \to \calM$ is defined by $\phi_Y (Q) \coloneqq \phi(Y, Q)$.

For a manifold $\calM$, let $\gothX(\calM)$ denote the set of smooth vector fields on $\calM$.
Let $f : \vecL \to \gothX(\calM)$, $u \mapsto f_u$, be a dynamical system on the manifold $\calM$ with input space $\vecL$.
Let $\grpG$ act on $\calM$ by the action $\phi: \grpG \times \calM \to \calM$.
For any $Y \in \grpG$, define the vector field pushforward of $f_u \in \gothX(\calM)$ by $Y$ by
\begin{align*}
\phi_\ast (Y, f_u) \coloneqq \tD \phi_Y \circ f_u \circ \phi_{Y^{-1}} \in \gothX(\calM).
\end{align*}
The system $f$ is said to be \textit{equivariant} if there exists an input group action $\psi: \grpG \times \vecL \to \vecL$ such that
\begin{align*}
\phi_\ast (Y, f_u) = f_{\psi(Y, u)}
\end{align*}
for all $Y \in \grpG, u \in \vecL$.
In other words, Figure \ref{fig:equivariant3} commutes.
Note that this also implies that if $f$ is equivariant \citep{mahonyObserverDesignNonlinear2022}, then for all $Q \in \calM, Y \in \grpG$, we have
\begin{align}
\tD  \phi_Y \left[f_u(Q) \right] &= \tD  \phi_Y \left[f_u(\phi_{Y^{-1}} (\phi_Y (Q))) \right]\notag \\
    &= \phi_\ast (Y, f_u) (\phi_Y (Q)) \notag \\
    &= f_{\psi(Y, u)}(\phi_Y(Q))
\label{eq:equivariance}.
\end{align}

\begin{figure}
\centering
\begin{tikzcd}[row sep=large,column sep=huge]
\vecL \arrow{r}{f} \arrow[swap]{d}{\psi_Y} & \gothX(\calM) \arrow{d}{\phi_\ast (Y,\cdot) } \\%
\vecL \arrow{r}{f}& \gothX(\calM)
\end{tikzcd}
\caption{Equivariance as seen as a transformation of the input space commuting with the transformation of vector fields $\phi_\ast(Y,f_u) := \tD \phi_Y \circ f_u \circ \phi_{Y^{-1}}$ induced by the group action $\phi$.}
\label{fig:equivariant3}
\end{figure}

\begin{definition}\label{def:navigation_function}
A Morse function \citep{shastri2011elements} is a function $\Upsilon : \calM \to \R$ for which every critical point $p \in \calM$ is regular: that is, for any point $p$ such that $\tD \Upsilon(p) = 0$ and any chart $\varphi: U \to \R^n$ (with $p \in U \subset \calM$ and $\varphi(p) = 0$), then the Hessian
\begin{align*}
    \frac{\partial^2 (\Upsilon \circ \varphi^{-1})}{\partial x^i \partial x^j}(0)_{i,j = 1,...,n}
\end{align*}
is non-degenerate.
A \textit{navigation function} \citep{Koditschek1989TheAO} is a proper Morse function with a unique local minimum.
In this way, these functions serve as suitable coordinate-free definitions of potentials.
\end{definition}

\section{Lie-Poisson Systems}
\label{sec:lp-systems}

Let $\grpG$ be a matrix Lie group and let $\gothg^\ast \cong \tT^\ast_{\Id} \grpG$ be the dual space of the Lie algebra $\gothg$.
Recalling \eqref{eq:ad_star}, consider a matrix ODE of the form
\begin{align}
    \begin{pmatrix}
        \dot{Q} \\
        \dot{P}
    \end{pmatrix} = \begin{pmatrix}
        Q V \\ \ad_{V}^\ast P + \tau
    \end{pmatrix} \label{eq:system},
\end{align}
where $(Q, P) \in \grpG \times \gothg^\ast$, $V \in \gothg$, and $\tau \in \gothg^\ast$.

Let $h: \gothg^\ast \times \R \to \R$ be a (possibly time-varying) function defined on the dual of the Lie algebra.
We will term this function the \textit{Hamiltonian} \citep{marsden1994introduction}.
Recalling \eqref{eq:algebra_ident} and identifying
\begin{align}\label{eq:VisddP}
V \coloneqq \ddP{h} \in \gothg,
\end{align}
then $(Q,P)$ correspond to the configuration and momentum states of a mechanical system and $\tau \in \gothg^\ast$ corresponds to an exogenous force.
In the mechanics literature, the uncontrolled dynamics $\dot{P} = \ad^\ast_{V} P$ (written as $\dot{P} = \ad^\ast_{\ddP{h}} P$) are referred to as the \textit{Lie-Poisson equation} \citep{marsden1994introduction} and the kinematics $\dot{Q} = Q V$ (written as $\dot{Q} = Q \ddP{h}$) are referred to as the \textit{reconstruction equation} \citep{cendraVariationalPrinciplesLie2003}.
If one defines a left-invariant Hamiltonian $H : \tT^\ast \grpG \to \R$ by $H(Q, \mu) := h(\tL^\ast_{Q} \mu)$ on the cotangent bundle, then these equations are equivalent to Hamilton's equations on $\tT^\ast \grpG$ \hbox{(i.e. $\dot{Q} = \pp{H(Q, \mu)}{\mu}, \dot{\mu} = -\pp{H(Q, \mu)}{Q}$)} \citep{cendraVariationalPrinciplesLie2003}.
In this paper, we emphasise that the correspondence $V \coloneqq \ddP{h}$ has the form of an additional constraint applied to a more general system \eqref{eq:system}.
In \S\ref{sec:extended_lp_systems} we will study the geometry of the generalised Lie-Poisson systems \eqref{eq:system} without the constraint \eqref{eq:VisddP}.

\begin{example}[Dynamics on $\SO(3)$]\label{ex:so3}
The Lie group $\SO(3)$ is the matrix Lie group defined by
\begin{align*}
    \SO(3) = \left\{ R \in \R^{3 \times 3} : R^\top R = \Id, \det(R) = 1 \right\}.
\end{align*}
The associated Lie algebra $\so(3)$ is the set of real $3 \times 3$ skew-symmetric matrices
\begin{align*}
    \so(3) = \left\{ U \in \R^{3 \times 3} : U + U^\top = 0 \right\}.
\end{align*}
Define the vector space isomorphism $(\cdot)^\times : \R^3 \to \so(3)$ by
\[
\begin{pmatrix}
    x_1 \\ x_2 \\ x_3
\end{pmatrix}^\times = \begin{pmatrix}
    0 & -x_3 & x_2 \\ x_3 & 0 & -x_1\\ -x_2 & x_1 & 0
\end{pmatrix},
\]
and write the inverse map as $(\cdot)^\vee : \so(3) \to \R^3$.
The Lie bracket can be computed
\begin{align*}
 [u^\times, v^\times] = u^\times v^\times - v^\times u^\times = (u \times v)^\times,
\end{align*}
and so imbuing $\R^3$ with the bracket $[u, v] = u \times v$ makes it isomorphic to $\so(3)$.
The dual space $\so^\ast (3)$ may also be identified with $\R^3$ (and $\so(3)$) via the dot product on $\R^3$: if $P \in \so^\ast (3)$, then it is identified with $\pi \in \R^3$ by $\langle P, u^\times \rangle = \pi^\top u$.
Thus, as an operator on $\R^3$, $\ad_u = u^\times$ and $\ad^\ast_u = (u^\times)^\top = -u^\times$.
Additionally, $\Ad_R \omega^\times = (R \omega)^\times \cong R \omega$, so $\langle \Ad^\ast_R P, \omega^\times \rangle = \langle P, \Ad_R \omega^\times \rangle \cong \pi^\top R \omega = (R^\top \pi)^\top \omega$, and $\Ad^\ast_R P \cong R^\top \pi$.

With these identifications, the Lie-Poisson equations on $\SO(3) \times \R^3$ become
\begin{subequations}
    \label{eq:rigid_body}
    \begin{align}
    \dot{R} &= R \omega^\times\\
    \dot{\pi} &= -\omega \times \pi + \tau,
\end{align}
\end{subequations}
where $\omega \coloneqq \dd{h}{\pi}$, for some $h : \R^3 \to \R$.
For a rigid-body, the kinetic energy Hamiltonian is given by $h(\pi) = \frac{1}{2} \pi^\top \inert^{-1} \pi$ \citep{marsden1994introduction},
where $\inert$ is the distance-weighted volume integral of the mass-density, $\inert = \int_{\mathcal{B}}  \rho(x)(\| x \|^2 \Id - x x^\top) \td^3 x$, and assumed to be non-degenerate.
Then $\omega = \dd{h}{\pi} = \inert^{-1}[\pi]$, $\pi = \inert \omega$, and the dynamics \eqref{eq:rigid_body} become the well-known (controlled) Euler equations for a rigid body \citep{marsden1994introduction}:
\begin{subequations}
    \label{eq:rigid_body2}
    \begin{align}
    \dot{R} &= R \omega^\times\\
    \inert \dot{\omega} &= -\omega \times \inert \omega + \tau.
\end{align}
\end{subequations}
\end{example}

\subsection{Energy shaping control of systems with invariant kinetic energy Hamiltonians}
\label{sec:stabilisation}
In this work, we are interested in the stabilisation and tracking control for a system \eqref{eq:system}.
We consider systems with kinetic energy Hamiltonian of the form
\begin{align}
    h(P, t) = \frac{1}{2}\langle P, \inert_t^{{-1}}[P]\rangle \label{eq:kinetic_energy},
\end{align}
where $\inert_t : \gothg \times \gothg \to \R$ is a time-parameterised, nondegenerate, positive definite (0, 2) \textit{inertia tensor}.
Here $\inert_t^{-1}$ is the inverse of the intrinsic mapping\footnote{
Any $(0,2)$-tensor $\calT$ induces a map $\calT^\flat : \vecV \to \vecV^\ast$, defined by $\langle \calT^\flat[x], y \rangle \coloneqq \calT(x, y)$.
The inverse of this map is often denote $\calT^\sharp = (\calT^\flat)^{-1}$.
To minimize notational complexity, we will simply write $\inert_t : \gothg \to \gothg^\ast$ and $\inert_t^{-1} : \gothg^\ast \to \gothg$.
}
$\inert_t : \gothg \to \gothg^\ast$ induced by $\inert_t$.
In this case, the velocity is given by $V = \ddP{h} = \inert^{-1}_t [P] \in \gothg$.
In general, because of the self-adjoint symmetry of $\inert$, we may write down $\inert_t = S_t^\ast \inert_0 S_t$ for some invertible $S_t : \gothg \to \gothg$ with $S_0 = \Id$.
To see this, choose a basis for $\gothg$ and then consider $\inert_t$ and $\inert_0$ as matrices and set $S(t) = \inert^{-1/2}_0 \inert_t^{1/2}$.
Since $\dot{\inert}_t = \dot{S}^\ast_t \inert_0 S_t + S^\ast_t \inert_0 \dot{S}_t$, straightforward computation yields
\begin{align*}
\dot{\inert}^{-1}_t
&= -\inert^{-1}_t \dot{S}^\ast_t S^{-\ast}_t - S^{-1}_t \dot{S}_t \inert^{-1}_t,
\end{align*}
where $S^{-\ast} = (S^\ast)^{-1}$.
Thus, the explicit time derivative of $h$ can be computed as
\begin{align}
    \frac{1}{2}\langle P, &\dot{\inert}^{^{-1}}_t \left[P\right] \rangle \notag\\
&= \frac{1}{2}\langle P, (-\inert^{-1}_t \dot{S}^\ast_t S^{-\ast}_t - S^{-1}_t \dot{S}_t \inert^{-1}_t ) \left[P\right] \rangle \notag \\
&= -\frac{1}{2}\langle P, \inert^{-1}_t \dot{S}^\ast_t S^{-\ast}_t [P] \rangle - \frac{1}{2}\langle P, S^{-1}_t \dot{S}_t \inert^{-1}_t \left[P\right] \rangle \notag \\
&= -\frac{1}{2}\langle \dot{S}^\ast_t S^{-\ast}_t [P], \inert^{-1}_t [P] \rangle - \frac{1}{2}\langle P, S^{-1}_t \dot{S}_t \inert^{-1}_t \left[P\right] \rangle \notag \\
&= -\frac{1}{2}\langle \dot{S}^\ast_t S^{-\ast}_t [P], \inert^{-1}_t [P] \rangle - \frac{1}{2}\langle \dot{S}^\ast_t S^{-\ast}_t P, \inert^{-1}_t \left[P\right] \rangle \notag \\
&= - \langle \dot{S}^\ast_t S^{-\ast}_t [P], \inert^{-1}_t [P] \rangle \label{eq:S_sym},
\end{align}
where we have used the fact that $\langle U, \inert^{-1}_t P \rangle = \langle P, \inert^{-1}_t U \rangle$ due to the self-adjoint symmetry of $\inert^{-1}$.

It is well-known that if $h$ has no explicit time-dependence, then the Lie-Poisson system \eqref{eq:system} can be stabilised to an arbitrary configuration $Q_0$ by energy shaping with a potential $\Upsilon$ and damping injection by some dissipative Rayleigh tensor $\mathcal{R}$ \citep{bulloGeometricControlMechanical2005,Koditschek1989TheAO,maithripalaAlmostglobalTrackingSimple2006}.
The case where $h$ has an explicit time-dependence is handled similarly.
\begin{theorem}
\label{thm:stability_time_var}
Consider the Lie-Poisson dynamics \eqref{eq:system} corresponding to a Hamiltonian $h$ of the form $h = \frac{1}{2}\langle P , \mathbb{I}_t^{{-1}} [P]\rangle = \frac{1}{2}\langle P , S^{-1}_t \inert_0^{{-1}} S^{-\ast} [P]\rangle$, where $\inert_t : \gothg \times \gothg \to \R_+$ is a nondegenerate positive-definite family of bilinear forms that is uniformly bounded above and below, with $\dot{\inert}_t$ bounded.
Let $\Upsilon : \grpG \to \R_+$ be a navigation function with minimum $Q_0 \in \grpG$ (Def.~\ref{def:navigation_function}) and let $\mathcal{R}$ be a positive-definite $(0, 2)$ tensor.
Define
\begin{subequations}
    \label{eq:control}
\begin{align}
    \tau_{\mathrm{ff}} &\coloneqq \dot{S}^\ast_t S^{- \ast}_t [P], \label{eq:control_tau_ff} \\
    \tau_{\mathrm{pd}} &\coloneqq - \mathcal{R}\left[V\right] -  \tL^\ast_Q \td \Upsilon (Q).
\end{align}
\end{subequations}
Then, the closed-loop system \eqref{eq:system} with control $\tau = \tau_{\mathrm{ff}} + \tau_{\mathrm{pd}}$ is locally asymptotically stable to $(Q_0, 0) \in \grpG \times \gothg^\ast$.
\end{theorem}

\begin{proof}
Define the Lyapunov function
\begin{align*}
    \calL(Q, P, t) = h(P, t) + \Upsilon(Q),
\end{align*}
which is positive definite and has a minimum at $(Q,P) = (Q_0, 0)$.

The time derivative $\dot{\calL}$ is given by
\begin{align*}
\dot{\calL} & = \dot{h}(P, t) + \dot{\Upsilon}(Q) \\
& = \td h [\dot{P}]  + \frac{\partial h}{\partial t} + \left\langle \td \Upsilon(Q), \dot{Q} \right\rangle\\
& = \left\langle \dot{P}, \ddP{h} \right\rangle + \frac{\partial h}{\partial t} + \left\langle \td \Upsilon(Q), \dot{Q} \right\rangle\\
& = \left\langle \tau, V \right\rangle + \frac{\partial h}{\partial t} + \left\langle \td \Upsilon(Q), QV \right\rangle,
\end{align*}
where we have used the fact that $\langle V, \ad^\ast_V P \rangle = \langle [V, V], P \rangle = 0$.
Substituting in the control law $\tau = \tau_{\mathrm{ff}} + \tau_{\mathrm{pd}}$, the closed-loop $P$ dynamics become
\begin{align}
\dot{P} & = \ad^\ast_{V} P - \mathcal{R}[V]
- \tL^\ast_Q \td \Upsilon (Q) + \dot{S}^\ast_t S^{- \ast}_t [P]. \label{eq:closed_loop}
\end{align}
Computing the time-derviative of $\calL$ along trajectories of the system, one obtains
\begin{align*}
    \dot{\calL}
& = \left\langle \tau_{\mathrm{pd}}, V \right\rangle + \left\langle \tau_{\mathrm{ff}}, V \right\rangle                                                                      \\
& \qquad \qquad + \pp{h}{t}  + \left\langle \td \Upsilon(Q), Q V \right\rangle                                                                                       \\
& = \left\langle - \mathcal{R} \left( V\right) -  \tL^\ast_Q \td \Upsilon(Q), V \right\rangle                                                                          \\
& \qquad \left\langle \dot{S}^\ast_t S^{- \ast}_t [P] , V\right\rangle + \pp{h}{t}  +  \left\langle \td \Upsilon(Q), Q V \right\rangle \\
& = -\mathcal{R} \left( V, V \right) + \left\langle \dot{S}^\ast_t S^{- \ast}_t [P], V \right\rangle + \pp{h}{t}                                                                                                             \\
& \qquad -  \left\langle \tL^\ast_Q \td \Upsilon(Q),  V  \right\rangle +  \left\langle \td \Upsilon(Q),  \tL_Q V \right\rangle                                        \\
& = -\mathcal{R} \left( V, V \right),
\end{align*}
which is negative semi-definite.
Since $\Upsilon$ is a Morse function and $Q_0$ is a local minimum, there exists a compact set $\mathcal{D} \subset \grpG$ containing $Q_0$ and an upper bound $B > 0$ such that $0 \leq \Upsilon(Q) \leq B$ for all $Q \in \mathcal{D}$.
The bilinear form $\inert_t$ is uniformly bounded above and below in time, so $\inert_t^{-1}$ is also uniformly bounded above and below in time, and there exist constants $C_1, C_2 > 0$ such that
\begin{align*}
    C_1 \| P \|^2 \leq h(P, t) \leq C_2 \| P \|^2
\end{align*}
for all $t \geq 0$.
Then $ C_1 \| P \|^2 \leq \calL(Q, P, t) \leq B + C_2 \| P \|^2$ for all $(Q, P) \in \mathcal{D} \times \gothg^\ast$ and $t \geq 0$.
Then, by the Lasalle-Yoshizawa theorem \citep[Theorem 8.4]{khalilNonlinearSystems2002}, there exists a compact set $\mathcal{B} \subset \mathcal{D} \times \gothg^\ast$ such that if $(Q_0, P_0) \in \mathcal{B}$, then the trajectory $(Q(t), P(t))$ is bounded (and thus exists for all time) and satisfies $\mathcal{R}(V, V) \to 0$.
Thus, $V \to 0$ and because $V = \inert^{{-1}}_t[P]$ and $\inert_t$ is nondegenerate, it also holds that $P \to 0$.
The map $\dot{P}$ \eqref{eq:closed_loop} is continuous on the compact set $\mathcal{B}$ and so by the Heine-Cantor theorem \citep[Theorem 4.19]{rudin1976principles}, $\dot{P}$ is uniformly continuous along trajectories of the system.
Thus, by Barbalat's lemma \citep[Lemma 8.2]{khalilNonlinearSystems2002}, $P(t) \to 0$ implies that $\dot{P}(t) \to 0$.
Examining the closed-loop equations \eqref{eq:closed_loop}, this is only possible if $- \tL^\ast_Q \td \Upsilon(Q) \to 0$, which ensures that that $Q \to Q_0$.
\end{proof}

\begin{remark}
Due to the global topological structure of the Lie groups and the properties of the navigation functions used, the basin of attraction is typically very large.
For example, on $\SO(3)$ the stability is almost global \citep{chaturvediRigidBodyAttitudeControl2011}.
\end{remark}

\begin{example}[Stabilising control on $\SO(3)$]\label{ex:so3_stab}
\makebox[0.5cm]{}
\newline 
Consider dynamics on $\SO(3)$ as introduced in Example~\ref{ex:so3}.
A navigation function with minimum $I \in \SO(3)$ is given by \citep{Koditschek1989TheAO}
\begin{align*}
\Upsilon_{\SO(3)}(R(\theta, n)) & = \tr(K(\Id - R(\theta, n)))  \notag       \\
& = -(1-\cos\theta)\tr(K (n^\times)^2)
\end{align*}
where $K \in \R^{3 \times 3}$ is a diagonal symmetric matrix whose eigenvalues satisfy
$\lambda_i + \lambda_j > 0$ for $i \not= j \in \{1, 2, 3\}$ \citep{bulloGeometricControlMechanical2005}.
The eigenvalue condition ensures the Hessian of $\Upsilon$ is positive definite on $\so(3) \equiv \tT_I \SO(3)$.
Note that if $K = \Id$ then $\Upsilon_{\SO(3)}(R(\theta, n)) = 1 - \cos(\theta)$.
Define a Rayleigh dissipation by $\calR(v, w) := v^\top K_d w$, where $K_d \in \R^{3 \times 3}$ is positive-definite.
To generate the control term one must compute $\td \Upsilon_{\SO(3)} \in \tT^\ast \SO(3)$.
For arbitrary $u^\times \in \so(3)$, one has
\begin{align}
\langle \tL_R^\ast \td \Upsilon_{\SO(3)}, u^\times \rangle
& = \langle \td \Upsilon_{\SO(3)}, R u^\times \rangle \notag \\
&= -\tr(K_p R u^\times) \notag \\
& = -\tr\left(\frac{1}{2}(K_p R - R^\top K_p^\top) u^\times\right) \notag \\
& = u^\top \left(K_p R - R^\top K_p^\top\right)^\vee, \label{eq:vee_RK}
\end{align}
and so $\gothg^\ast \ni \tL_R^\ast \td \Upsilon_{\SO(3)} \cong \left(K_p R - R^\top K_p^\top\right)^\vee \in \R^3$.
Note that for the standard kinetic energy Hamiltonian, the inertia $\mathbb{I}$ is constant for and $\dot{S} = 0$.
It follows that $\tau_{\text{ff}} = 0$ \eqref{eq:control_tau_ff}.
Define
\begin{align*}
    \tau = -\left(K_p R - R^\top K_p^\top\right)^\vee - K_d \omega.
\end{align*}
This corresponds to the control proposed in Theorem~\ref{thm:stability_time_var} and it follows that the system is stabilised to the identity $\Id$ locally.
The proposed control corresponds to the standard control from the literature \citep{chaturvediRigidBodyAttitudeControl2011} and the system is known to be almost globally asymptotically stable and locally exponentially stable.
\end{example}

\section{Extended Lie-Poisson Systems}
\label{sec:extended_lp_systems}

We are interested in extending the stabilisation results of \S \ref{sec:stabilisation} to the tracking of smooth trajectories.
To this end we will study a novel form of the Lie-Poisson equation that we term the \emph{extended Lie-Poisson} equation and show that these equations are left invariant with respect to a group action on the cotangent bundle.
This will lead to structure that we exploit to define an equivariant error and solve the tracking problem.

The \emph{extended Lie-Poisson} system is defined to be the same ordinary differential equation \eqref{eq:system} without the identification \eqref{eq:VisddP}.
That is, where the velocity $V$ is treated as a free input decoupled from the momentum $P$.
The resulting system no longer corresponds directly to physical dynamics, however, as we see below it admits a natural symmetry.
The total input space is given by $\gothg \times \gothg^\ast$ where $(V, \tau) \in \gothg \times \gothg^\ast$ are \emph{free} inputs.
The behaviour of the extended Lie-Poisson system contains the Lie-Poisson dynamics as a sub-behaviour of the system, obtained by assigning $V = \ddP{h}$.

Define a general input $U \coloneqq (V, \tau) \in \gothg \times \gothg^\ast$, a state $X \coloneqq (Q, P) \in \grpG \times \gothg^\ast$, and a \textit{system function} $f : \gothg \times \gothg^\ast \to \gothX (\grpG \times \gothg^\ast)$,
\begin{align}
    f_{U}(X) \coloneqq f(U,X) := (Q V , \ad_V^\ast P + \tau) \label{eq:f_def}.
\end{align}
Then the extended system can be written in compact notation as \hbox{$\dot{X} = f_U(X)$}.

The state-space,  $\grpG \times \gothg^\ast$,  of the extended Lie-Poisson system can be provided with a semi-direct product structure that makes it into a Lie group
$\Gstar = \grpG \ltimes \gothg^\ast$ in its own right  \citep{holmEulerPoincareEquations1998,engo-monsenPartitionedRungeKutta2003,jayaramanBlackScholesTheoryDiffusion2020}.
Group multiplication is given by
\begin{align}
    (Q_1, P_1) (Q_2, P_2) = (Q_1 Q_2, \Ad_{Q_2}^\ast P_1 + P_2), \label{eq:group}
\end{align}
for $Q_1, Q_2 \in \grpG$ and $P_1, P_2 \in \gothg^\ast$.
The identity is given by $\Id = (\Id, 0)$ and the inverse operation is given by \hbox{$(Q, P)^{-1} = (Q^{-1}, -\Ad_{Q^{-1}}^\ast P)$}.
Note that while $\grpG^\ltimes_{\gothg^\ast}$ is a semi-direct product of matrix Lie groups, it is not a matrix Lie group itself as written, although a matrix representation of the action could be obtained by selecting a basis of $\gothg$.
The group multiplication can be motivated by noting that it is the combination of the coadjoint action $\Ad^\ast_{Q}$ (which acts canonically on $\gothg^\ast$) and a momentum translation.
The uncontrolled dynamics of the Lie-Poisson system are known to be invariant on orbits of the coadjoint action, a property that is related to the left-invariant structure of the extended Lie-Poisson system that we demonstrate below.

The Lie algebra of $\Gstar$ is given by the vector space
\begin{align*}
    \gothg^\ltimes_{\gothg^\ast} \coloneqq \tT_\Id (\grpG \times \gothg^\ast) \cong  \gothg \ltimes \gothg^\ast,
\end{align*}
with Lie bracket (Appendix \ref{sec:app_semidirect})
\begin{align*}
    [(V_1, P_1), (V_2, P_2)] = ([V_1, V_2], -\ad^\ast_{V_1} P_2 + \ad^\ast_{V_2} P_1 ).
\end{align*}
The adjoint map $\Ad : \Gstar \times \gstar \to \gstar$ is given by (Appendix \ref{sec:app_semidirect})
\begin{align*}
    \Ad_X U    & = \tD \tR_{X^{-1}}[\tD \tL_X[U]] \notag                                                \\
               & = (\Ad_{Q} V , \Ad^\ast_{Q^{-1}} (\ad^\ast_{V}  P + \tau)).
\end{align*}

\begin{theorem}
\label{prop:left-eq}
The extended Lie-Poisson \eqref{eq:system} equations are left-invariant in the semi-direct group structure $\grpG \ltimes \gothg^\ast$.
That is, recalling \eqref{eq:f_def},
\[
f_{U}(X) = \tD \tL_{X} [U],
\]
for $X \in \Gstar$, where $U \coloneqq (V,\tau) \in \gothg^\ltimes_{\gothg^\ast}$ lies in the algebra of $\Gstar$.
\end{theorem}

\begin{proof}
    From the definition of $U$ and equation \eqref{eq:DL}, we have
    \begin{align}
        \tD \tL_{X}[U]
         & = (Q V, \ad_{V}^\ast P + \tau) \label{eq:diff_left_trans}.
    \end{align}
    The result follows by comparing \eqref{eq:f_def} and \eqref{eq:diff_left_trans}.
\end{proof}

Let $Y \in \Gstar$ be arbitrary.
Then (cf.~Figure~\ref{fig:poissonEquivariant})
\begin{align*}
    \tD \tR_Y f_u(X) = \tD \tR_Y \tD \tL_X U &= \tD \tL_X \tD \tL_Y \tD \tL_{Y^{-1}} \tD \tR_Y U \\
                      & = \tD \tL_{XY} \Ad_{Y^{-1}} U                                    \\
                      & = f_{\Ad_{Y^{-1}} U}(\tR_Y X).
\end{align*}
In particular, the system \eqref{eq:system} is equivariant \eqref{eq:equivariance} with respect to the right-translation $\tR$ (state space $\phi$) and adjoint $\Ad$ (input $\psi$) group actions.
\begin{figure}
\centering
\begin{tikzcd}[row sep=large,column sep=huge]
\gstar \arrow{r}{f} \arrow[swap]{d}{\Ad_{Y^{-1}}} & \gothX(\Gstar) \arrow{d}{\tR_{\ast}(Y, \cdot)} \\%
\gstar \arrow{r}{f}& \gothX(\Gstar)
\end{tikzcd}
\caption{Equivariance of extended Lie-Poisson systems.}
\label{fig:poissonEquivariant}
\end{figure}
Note that the Adjoint operator $\Ad_{Y^{-1}} U$  here is taken on $\Gstar$ since $Y \in \Gstar$.

\section{Error Systems}

Let $\grpG \subset \R^{d\times d}$ be a $m$-dimensional matrix Lie group and $\gothg^\ast \subset \R^{d\times d} $ be the associated Lie coalgebra.
Let $h : \gothg^\ast \to \R$ be a Hamiltonian.
Consider a desired trajectory $(Q_d, P_d, V_d, \tau_d) : \R \to \grpG \times \gothg^\ast$ that satisfies of the Lie-Poisson system
\eqref{eq:system} with the constraint \eqref{eq:VisddP}.
Let $X(t) = (Q, P)$ be a system trajectory evolving according to the extended Lie-Poisson equations \eqref{eq:system}.
Specifically, we do not impose that \eqref{eq:VisddP} holds for the trajectory $X(t)$ in the following analysis, although the constraint will be applied during the control design undertaken later.
We wish to construct an error between $X(t)$ and $X_d(t)$.
In prior work, the most commonly used error is the left-invariant error \citep{murray1995proportional,leeGeometricTrackingControl2010,chaturvediRigidBodyAttitudeControl2011}
\begin{align*}
E_\text{prior} = (Q_d^{-1} Q, V - \Ad_{Q^{-1}Q_d}V_d),
\end{align*}
that is derived from the left-trivialised tangent group structure on $\grpG \times \gothg$ \citep{murray1995proportional}.
Note that here the Adjoint used is defined on the group $\grpG$ not on $\Gstar$.
The error system dynamics for $E_\text{prior}$ are not an extended Lie-Poisson system.
They are certainly not a Lie-Poisson system, and there is no underlying Hamiltonian that generates these dynamics.
The most common approach to control \citep{bulloTrackingFullyActuated1999} uses feedforward compensation to impose a certain form of passivity by cancelling non-linearities and then designs a stabilizing control for the compensated system.

In this document, we pose the error directly on $\grpG \ltimes \gothg^\ast$ (acting on the state $(Q,P)$ rather than $(Q,V)$) and use the group multiplication \eqref{eq:group} to generate the error.
In addition, we use a right-invariant Lie group error motivated by recent work in equivariant observer design \citep{mahonyObserverDesignNonlinear2022,vangoorEquivariantFilterEqF2022}
\begin{align}
    (Q_E,P_E) \coloneqq E & = X X_d^{-1} = \phi(X_d^{-1}, X)\notag                              \\
                    & = (Q Q_d^{-1}, \Ad^\ast_{Q_d^{-1}} (P - P_d)) \label{eq:error_def}.
\end{align}
The motivation for this choice is that left-invariant systems are \textit{right-equivariant} and the right-invariant error is synchronous \citep{vangoorSynchronousModelsFundamental2025}.

\begin{lemma}
    \label{lemma:error_dynamics}
Let $X(t) = (Q(t), P(t))$ and $X_d(t) = (Q_d(t), P_d(t))$ be the system and desired trajectories satisfying \eqref{eq:system}
for inputs $U = (V,\tau), U_d = (V_d, \tau_d) \in \gothg \times \gothg^\ast$.
Define $E = (Q_E, P_E)$ to be the control error \eqref{eq:error_def}.
Then the error dynamics have the form of an extended Lie-Poisson system
\begin{align}
\dot{Q}_E & = Q_E V_E, & \dot{P}_E & = \ad_{V_E}^\ast P_E + \tau_E, \label{eq:error_sys}
\end{align}
where
\begin{align}
U_E = (V_E, \tau_E) \coloneqq & \Ad_{X_d} (V - V_d, \tau - \tau_d) \label{eq:error_input}.
\end{align}
\end{lemma}

\begin{proof}
Compute
\begin{align*}
\dot{E} & = \tD \tR_{X_d^{-1}} \dot{X} - \tD \tL_X \tD \tL_{X^{-1}_d} \tD \tR_{X^{-1}_d} \dot{X}_d    \\
& = \tD \tR_{X^{-1}_d} f_U (X) - \tD \tL_X \tD \tL_{X^{-1}_d} \tD \tR_{X^{-1}_d} f_{U_d}(X_d) \\
& = f_{\Ad_{X_d} U} (R_{X^{-1}_d} X )                                         \\
& \qquad \qquad - \tD \tL_{X X^{-1}_d} f_{\Ad_{X_d} U_d}(R_{X^{-1}_d} X_d)     \\
& = f_{\Ad_{X_d} U} (E)- \tD \tL_{E} f_{\Ad_{X_d} U_d}(\Id)                 \\
& = f_{\Ad_{X_d} U} (E) - f_{\Ad_{X_d} U_d}(E)                            \\
& = f_{\Ad_{X_d} (U-U_d)} (E)\\
& = f_{U_E}(X_E).
\end{align*}
\end{proof}

\section{Tracking Control for Lie-Poisson Systems}
\label{sec:tracking_control}

Lemma~\ref{lemma:error_dynamics} shows that the error dynamics \eqref{eq:error_sys} are associated with an extended Lie-Poisson system.
To show that this system is in fact Lie-Poisson then we must find a Hamiltonian $h_E(P_E,t)$ such that $V_E = \dd{h_E}{P_E}$.
If we reinstate the constraint $V = \ddP{h}$ on the desired trajectory (that is, that we assume the dynamics of the original system are Lie-Poisson), then this is possible.
This construction will lead to the definition of an ``error kinetic energy'' that forms the key element underlying the proposed tracking control design.

\begin{lemma}
\label{thm:error_Hamiltonian}
Let $(Q_d(t), P_d(t)): \R \to \grpG \ltimes \gothg^\ast$ be a solution of \eqref{eq:system} for a kinetic energy Hamiltonian $h = \frac{1}{2}\langle P, \inert^{{-1}} P \rangle$, with input $\tau_d: \R \to \gothg^\ast$.
Define a time-varying tensor
\[
    \bar{\inert}_t (V, W) \coloneqq  \inert( \Ad_{Q_d^{-1}(t)} V, \Ad_{Q_d^{-1}(t)} W).
\]
Define $h_E : \gothg^\ast \to \R$ by
\[
h_E(P_E) := \frac{1}{2} \langle P_E, \bar{\inert}^{-1}_t [P_E] \rangle
\]
to be the kinetic energy Hamiltonian defined by $\bar{\inert}_t$.
Assume both $V \coloneqq \ddP{h(P)}$ and $V_d \coloneqq \dd{h(P_d)}{P_d}$.
Then the error system dynamics $(Q_E, P_E)$ \eqref{eq:error_sys} are Lie-Poisson dynamics induced by $h_E$ with constraint $V_E = \dd{h_E}{P_E}$ and free input $\tau_E$.
\end{lemma}

\begin{proof}
From Lemma \ref{lemma:error_dynamics} the error dynamics are extended Lie-Poisson, so it remains to show that $V_E = \dd{h_E}{P_E}$ and verify $\tau_E$ is free.
The error velocity $V_E$ is given by $\Ad_{Q_d(t)}(V - V_d)$, the first component of \eqref{eq:error_input}.
Writing $V \coloneqq \ddP{h(P)} = \mathbb{I}^{-1} [P]$, $V_d \coloneqq \mathbb{I}^{-1} [P_d] = \dd{h(P_d)}{P_d}$ then $V_E = \Ad_{Q_d(t)}\inert^{-1}[P - P_d]$.
From \eqref{eq:error_def}, we may write $P-P_d = \Ad^\ast_{Q_d} P_E$, and so $V_E = \Ad_{Q_d(t)}\inert^{-1}[\Ad^\ast_{Q_d(t)} P_E] = \bar{\inert}^{-1}[P_E]$.
Thus, $h(P_E) = \frac{1}{2}\langle P_E, V_E \rangle$ and $V_E = \dd{h_E}{P_E}$.
For arbitrary $\tau_E$ one sets
\begin{align}
\tau := \tau_d + \Ad^\tau_{X^{-1}_d}(V_E, \tau_E)
\label{eq:system_input},
\end{align}
where $\Ad^\tau_{X_d}$ denotes the second ($\gothg^\ast$) component of $\Ad_{X_d}$, $V_E$ is defined as in \eqref{eq:error_input}.
By choosing $\tau$ in this manner, then any desired $\tau_E$ can be assigned, and this completes the proof.
\end{proof}

Note that the error inertia corresponds to a time-varying linear map $\bar{\inert}_t : \gothg \to \gothg^\ast$ that can be written
\[
    \bar{\inert}_t\coloneqq \Ad^\ast_{Q_d^{-1}(t)} \circ \inert \circ \Ad_{Q_d^{-1}(t)}.
\]
This is explicitly of the form $ S^\ast_t \inert S_t$, with $S_t = \Ad_{Q_d^{-1}(t)}$.
Then $\ddt Q_d(t) = Q_d(t) u_d(t)$ for some $u_d(t) : \R \to \gothg$, and one can compute \eqref{eq:S_sym}
\begin{align}
\dot{S}_t = - \ad_{u_d(t)} \circ \Ad_{Q_d^{-1}(t)} = - \Ad_{Q_d^{-1}(t)} \circ \ad_{\Ad_{Q_d(t)} u_d(t)} \label{eq:S_dot}.
\end{align}
In the following Theorem we will consider tracking control for a Lie-Poisson system.
Conceptually, we construct the error between true and desired systems exploiting the geometry of the extended Lie-Poisson system, and then restrict to only those trajectories where the velocity constraint \eqref{eq:VisddP} holds.
\begin{theorem}
\label{thm:tracking_control}
Consider true and desired trajectories $(Q,V), (Q_d, P_d) : \R \to \grpG \times \gothg^\ast$ that satisfy the Lie-Poisson system \eqref{eq:system} with the constraint \eqref{eq:VisddP} for a Hamiltonian $h : \gothg^\ast \to \R$ with inputs $\tau$ and $\tau_d$ respectively.
Assume $Q_d(t)$ is bounded, and $\Ad_{Q_d(t)}$ and $\Ad_{Q^{-1}_d(t)}$ are positively bounded above and below (see Remark \ref{remark:boundedness}).
Define the error \eqref{eq:error_def}.
Let $\Upsilon : \grpG \to \R_+$ be a navigation function with global minimum $(\Id , 0 ) \in \Gstar$ and let $\mathcal{R}$ be a positive-definite $(0, 2)$ tensor.
The closed-loop input $\tau$ is defined by \eqref{eq:system_input} where
$\tau_E := \tau_{E_{\mathrm{ff}}} + \tau_{E_{\mathrm{pd}}}$
with
\begin{subequations}
\label{eq:control_tracking}
\begin{align}
\tau_{E_{\mathrm{ff}}} &\coloneqq -\ad^\ast_{\Ad_{Q_d}U_d} P_E. \\
\tau_{E_{\mathrm{pd}}} &\coloneqq - \mathcal{R}\left[V_E\right] -  \tL^\ast_{Q_E} \td \Upsilon (Q_E).
\end{align}
\end{subequations}
Then the closed-loop error system \eqref{eq:error_sys} is locally asymptotically stable to $(\Id,0)$ and $Q(t) \to Q_d(t)$ and $P(t) \to P_d(t)$.
\end{theorem}

\begin{proof}
By Lemma \ref{thm:error_Hamiltonian}, the error system \eqref{eq:error_def} is a Lie-Poisson system.
Recalling the definition of $\tau$ \eqref{eq:system_input} then in the error system \eqref{eq:error_sys} one has $\tau_E = \tau_{E_{\mathrm{ff}}} + \tau_{E_{\mathrm{pd}}}$.

Using \eqref{eq:S_dot}, one computes
\begin{align*}
\dot{S}^\ast S^{-\ast}_t &= (- \Ad_{Q_d^{-1}(t)} \circ \ad_{\Ad_{Q_d(t)} u_d(t)})^\ast (\Ad_{Q_d^{-1}(t)})^{-\ast}\\
& = - \ad^\ast_{\Ad_{Q_d(t)} u_d(t)} \Ad^\ast_{Q_d^{-1}(t)} \Ad^\ast_{Q_d(t)}\\
& = - \ad^\ast_{\Ad_{Q_d(t)} u_d(t)}.
\end{align*}
Comparing \eqref{eq:control} and \eqref{eq:control_tracking}, we note that these control laws are identical up to the variable substitutions $Q \leftrightarrow Q_E, P \leftrightarrow P_E$, $V \leftrightarrow V_E$ and $\tau \leftrightarrow \tau_E$,
with $\dot{S}^\ast S^{-\ast}_t$ taking the explicit form given above.
The term $\dot{\inert}_t$ is uniformly bounded if $\Ad_{Q_d(t)}$ and $U_d(t)$ are uniformly bounded.
Thus, by Theorem \ref{thm:stability_time_var}, $Q_E \to \Id$ and $P_E \to 0$.
If $\Ad_{Q_d^{-1}(t)}$ is uniformly bounded below then $P_E \to 0$ implies $P \to P_d$, and if $Q_d(t)$ is bounded then $Q_E \to \Id$ implies $Q \to Q_d$.
\end{proof}

\begin{remark}
\label{remark:boundedness}
For compact Lie groups or Lie groups that are the semidirect product of a compact Lie group with a finite-dimensional vector space (for example, $\SO(2)$, $\SO(3)$, $\SE(2)$, $\SE(3)$, $\SE_2(3)$, etc), it suffices for $Q_d(t)$ to be bounded.
Then lower and upper boundedness of $\Ad_{Q_d(t)}$ and $\Ad_{Q_d^{-1}(t)}$ are guaranteed by the compact component of the Lie group.
\end{remark}

\section{Attitude control on $\SO(3)$}
\label{sec:so3}

In this example we consider tracking control for the $\SO(3)$ dynamics introduced in Examples~\ref{ex:so3} and \ref{ex:so3_stab}.
We translate the abstract formula presented in \S\ref{sec:lp-systems} and \S\ref{sec:extended_lp_systems} into concrete algebraic formula written in matrix/vector coordinates.
This both provides an example of how to apply the theory presented in the paper, and also provides a comparison with prior work.

Recall the notation introduced in Example~\ref{ex:so3}.
The semi-direct group product on $\SO(3) \ltimes \R^3 \cong \SO(3) \ltimes \so^\ast(3)$ can be written
\begin{align*}
    (R_1, \pi_1)(R_2, \pi_2) = (R_1 \tR_2, \tR_2^\top \pi_1 + \pi_2).
\end{align*}
The inverse element is $(R, \pi)^{-1} = (R^\top, -R\pi)$.
Right multiplication on $\SO(3) \ltimes \R^3$ (the $\phi$ action) is given by
\begin{align*}
    \tR_{\left(R_Y, \pi_Y\right)} (R, \pi) = (R, \pi)(R_Y, \pi_Y) = (R R_Y, R_Y^\top \pi + \pi_Y).
\end{align*}
Finally, the adjoint map (the $\psi$ action) $\Ad: (\SO(3) \ltimes \R^3) \times  (\R^3 \times \R^3) \to \R^3 \times \R^3$ is given by
\begin{align*}
    \Ad_{\left(R_Y, \pi_Y\right)} (\omega, \tau) = ( R_Y \omega, R_Y (-\omega \times \pi_Y) + R_Y \tau).
\end{align*}
where $\R^3 \times \R^3 \cong \so(3) \times \so^\ast(3)$.

The error state \eqref{eq:error_def} and inputs \eqref{eq:error_input} are given by
\begin{align*}
    (R_E, \pi_E) &= \tR_{(R_d^\top, -R_d \pi_d)}(R, \pi) = (R R_d^\top, R_d (\pi - \pi_d))\\
    (\omega_E, \tau_E) &= \Ad_{(R_d, \pi_d)}(\omega-\omega_d, \tau-\tau_d)\\
                       &= (R_d (\omega - \omega_d),\\
                       & \qquad -R_d(\omega - \omega_d) \times \pi_d + R_d(\tau - \tau_d) ).
\end{align*}
From Theorem \ref{thm:tracking_control} the \textit{tracking} control law for the error input is given by
\begin{align*}
    \tau_E &= -\ad^\ast_{\Ad_{R_d}\omega_d} \pi_E -  \tL^\ast_{R_E} \td \Upsilon (R_E) - \mathcal{R}\left[\omega_E\right]\\
           &=  (R_d \omega_d) \times \pi_E - (K_p R_E - R_E^\top K_p^\top)^\vee -K_d \omega_E.
\end{align*}
Define $\tilde{\omega} = \omega - \omega_d$, $\tilde{\tau} = \tau - \tau_d$.
Unwinding the definitions of $R_E, \pi_E, \omega_E$ and $\tau_E$, one has
\begin{align*}
R_d \tilde{\tau} &= R_d \tilde{\omega} \times \pi_d
-K_d R_d \tilde{\omega} - (K_p R_E - R_E^\top K_p^\top)^\vee\\
&\qquad\qquad\qquad  +  (R_d \omega_d) \times \pi_E.
\end{align*}
Then the closed loop input $\tau$ is recovered using \eqref{eq:system_input}:
\begin{align}
\tau^\text{EqT} &= \tau_d +  \tilde{\omega} \times \inert \omega_d +  \omega_d \times R_d^\top \pi_E \notag \\
&\qquad\qquad\qquad - R_d^\top K_d R_d \tilde{\omega} - R_d^\top (K_p R_E - R_E^\top K_p^\top)^\vee \notag \\
&= \tau_d +  \tilde{\omega} \times \inert \omega_d +  \omega_d \times  \inert \tilde{\omega} \notag \\
&\qquad - R_d^\top K_d R_d \tilde{\omega} - R_d^\top (K_p R_E - R_E^\top K_p^\top)^\vee. \label{eq:so3_control}
\end{align}
We will refer to this as the \emph{Equivariant Tracking} (EqT) control.

To provide context we compare \eqref{eq:so3_control} with the right-invariant error form of the tracking controller drawn from Bullo \etal~\citep{bulloGeometricControlMechanical2005}.
We will refer to the Bullo \etal~controller as a Geometric Tracking (GT) tracking control.
In \citep{bulloGeometricControlMechanical2005}, the velocity tracking error is defined by
\begin{align*}
    K(v_E) = \frac{1}{2} \mathbb{G}_{R(t)} (v_E, v_E),
\end{align*}
where the kinetic energy metric is defined as $\mathbb{G}_{R} \coloneqq \tD \tL_R^\ast \inert \tD \tL_R$ for $\inert$ the body-fixed inertia introduced in Example~\ref{ex:so3}, and $v_E = R (\omega^\times - \omega^\times_d) \in \tT_R \SO(3)$.
Note that in \citep{bulloGeometricControlMechanical2005}, $K(v_E)$ is used purely as a velocity error term in the appropriate Lyapunov function and there is no interpretation as kinetic energy of any error system.

Recalling Example~\ref{ex:so3_stab}, we use the navigation function $\Upsilon (R_E)$ as the configuration error function and note that it satisfies \citep[Proposition 11.31]{bulloGeometricControlMechanical2005}.
Applying the formulae in \citep{bulloGeometricControlMechanical2005} one obtains a control
\begin{align*}
\tau^\text{GT} = - K_d \tilde{\omega} - R_d^\top (K_p R_E - R_E^\top K_p^\top)^\vee + \inert (\nabla_\omega \omega_d + \dot{\omega}_d),
\end{align*}
where $\nabla_\omega \omega_d$ is the Levi-Civita connection corresponding to the kinetic energy metric.
The feedforward terms $\nabla_\omega \omega_d + \dot{\omega}_d$ arrive naturally by taking the covariant derivative of $v_E$ along $\dot{R}$ and left-trivialising.
For $\SO(3)$ with inertia $\inert$ one has
\begin{align*}
\nabla_{\omega} \omega_d \coloneqq \frac{1}{2} \omega \times \omega_d + \frac{1}{2} \inert^{-1} [\omega_d \times \inert \omega] + \frac{1}{2} \inert^{-1} [\omega \times \inert \omega_d].
\end{align*}
Expanding $\nabla_\omega \omega_d$ and $\dot{\omega}_d$, the GT control can be written
\begin{align}
\tau^\text{GT} &=  \tau_d  + \frac{1}{2} \inert (\omega \times \omega_d) + \frac{1}{2} \omega \times \inert \omega_d + \frac{1}{2} \omega_d \times \inert \omega \notag\\
& \quad - \omega_d \times \inert \omega_d - K_d \tilde{\omega} - R_d^\top (K_p R_E - R_E^\top K_p^\top)^\vee. \label{eq:bullo_control}
\end{align}

Consider the EqT \eqref{eq:so3_control} and GT \eqref{eq:bullo_control} control laws in comparison to each other.
The proportional feedback term $- R_d^\top (K_p R_E - R_E^\top K_p^\top)^\vee$ is the same in both control designs.
This should be expected since the same navigation function is used and this term depends only on the geometry of the group $\SO(3)$ and not on the cotangent or tangent bundle structure.
The damping term, $-R_d^\top K_d R_d \tilde{\omega}$, in the EqT compares to $-K_d \tilde{\omega}$ term in the GT control.
This difference can be explained by noting that the damping gain is defined in spatial coordinates for the EqT, and in body coordinates for the GT.
For a homogeneous gain matrix $K_d = \kappa_d I_3$ then there is no difference in the case of $\SO(3)$.

The feedforward terms vary from each other significantly and warrant a more detailed discussion.
The feedforward terms of the GT controller appear as a consequence of the covariant derivative of $\omega_d$ in the direction of $\omega$, while the feedforward terms in the EqT come from the time-varying inertia of the error system.
By inspection the feedforward terms for the EqT \eqref{eq:so3_control} and GT \eqref{eq:bullo_control} controls are
\begin{align*}
\tau^\text{EqT}_\text{ff} &:= \tau_d
+ \tilde{\omega} \times \inert  \omega_d
+  \omega_d \times  \inert \tilde{\omega} \\
\tau^\text{GT}_\text{ff} &:= \tau_d
+ \frac{1}{2} \tilde{\omega} \times \inert \omega_d
+ \frac{1}{2} \omega_d \times \inert \tilde{\omega}
+ \frac{1}{2} \inert (\omega \times \omega_d).
\end{align*}
In both the EqT and GT control, the Lyapunov function depends on the control through a left trivialised power term $\langle \tilde{\omega}, \tilde{\tau}\rangle = \tilde{\omega}^\top \tilde{\tau}$.
Computing this power term for the EqT and GT separately yields
\begin{align*}
\tilde{\omega}^\top (\tau^\text{GT}_\text{ff}& -\tau_d)\\
&= \tilde{\omega}^\top \left(\frac{1}{2} \tilde{\omega} \times \inert \omega_d
+ \frac{1}{2} \omega_d \times \inert \tilde{\omega}
+ \frac{1}{2} \inert (\omega \times \omega_d)\right)\\
&= \frac{1}{2} \tilde{\omega}^\top( \omega_d \times \inert \tilde{\omega})
+ \frac{1}{2} (\inert \tilde{\omega})^\top(\omega \times \omega_d)\\
&= \frac{1}{2} (\inert \tilde{\omega})^\top (\tilde{\omega} \times \omega_d )
+ \frac{1}{2} (\inert \tilde{\omega})^\top(\omega \times \omega_d)\\
&= (\inert \tilde{\omega})^\top (\omega \times \omega_d)
\end{align*}
and
\begin{align*}
\tilde{\omega}^\top (\tau^\text{EqT}_\text{ff} -\tau_d) &=
\tilde{\omega}^\top \left(\tilde{\omega} \times \inert  \omega_d
+  \omega_d \times  \inert \tilde{\omega} \right)\\
&=  \tilde{\omega}^\top (\omega_d \times  \inert \tilde{\omega})\\
&= (\inert \tilde{\omega})^\top (\tilde{\omega} \times \omega_d)\\
&= (\inert \tilde{\omega})^\top (\omega \times \omega_d).
\end{align*}
It follows that the feedforward terms for the EqT and GT controls correspond to the same instantaneous change in the the Lyapunov function.
Although instantaneously equivalent, the different passive gyroscopic forces in the feedforward terms do effect the trajectory of the system and lead to global differences in the convergence.
They also lead to significant differences in the magnitude of control action required.

Before discussing the difference between $\tau^\text{EqT}$ and $\tau^\text{GT}$ in more detail, there are two alternative controls to $\tau^{\text{EqT}}$ that are suggested by the analysis undertaken above.
Firstly, the gyroscopic terms could be removed entirely from the control, since they do not correspond to instantaneous change in energy of the error system:
\begin{align*}
    \tau^{\text{nog}} &= \tau_d +  \omega_d \times  \inert \tilde{\omega}\\
&\qquad - R_d^\top K_d R_d \tilde{\omega} - R_d^\top (K_p R_E - R_E^\top K_p^\top)^\vee.
\end{align*}
Secondly, in showing the correspondence between $\tau^{\text{GT}}$ and $\tau^{\text{EqT}}$ we proved that $\langle \tilde{\omega}, \omega_d \times \inert \tilde{\omega} \rangle
=
\langle \tilde{\omega},\inert (\omega \times \omega_d)\rangle $.
This suggests that the control
\begin{align*}
    \tau^{\text{asym}} &= \tau_d +  \tilde{\omega} \times \inert \omega_d + \frac{1}{2} \omega_d \times \inert \tilde{\omega} + \frac{1}{2} \inert (\omega \times \omega_d)\\
&\qquad - R_d^\top K_d R_d \tilde{\omega} - R_d^\top (K_p R_E - R_E^\top K_p^\top)^\vee
\end{align*}
is of interest to consider.
This correspondence is associated with the symmetrisation of the inertia \eqref{eq:S_sym} and can be applied to any Lie-Poisson system leading to a general control of the form $\tau^\text{asym}$.

All four of these controls are stabilising and are compared by Monte Carlo simulation in Figure \ref{fig:so3_error_orthogonal}.
The trajectory to be the tracked is the sinuoidal trajectory defined by $R_d(0) = \Id, \omega_d (0) = 0,\tau_d = (\cos(t), \sin(t), \cos(t) \sin(t))$.
For each iteration, the initial conditions for the actual system are perturbed to $R(0) = R_d(0) R(\theta, \phi, \psi)$, where the Euler angles $\theta, \phi, \psi$ are sampled uniformly from the interval $[-\pi, \pi]$.
The initial system angular velocity is perturbed by $\omega(0) = \omega_d(0) + \omega_p$, where $\omega_p$ is sampled from a normal distribution of mean $0$ and standard deviation $1.0$ in all directions.
The moment of inertia is chosen to be that of a model fixed-wing, which has light roll-yaw coupling \citep{burstonReverseEngineeringFixed2014}
\begin{align*}
    \inert = \begin{pmatrix}
        0.824 & 0 & 0.12 \\ 0 & 1.135 & 0 \\ 0.12 & 0 & 1.759
    \end{pmatrix}.
\end{align*}
The proportional gain is set to $K_p = \Id_3$ and the Rayleigh dissipation term $K_d = \frac{1}{2} \Id_3$ is chosen to be homogeneous, removing any difference due to damping.
The differences in the trajectories are purely due to the different choices of feedfoward controls.
Clearly, while the term $\tilde{\omega} \times \inert \omega_d$ is gyroscopic, it significantly effects the input magnitude required.
The EqT and GT controls appear to be superior to the two approximations discussed as variations.
It appears that the EqT control requires less control action for large error (the first 5 seconds of the response) but leads to slightly slower attitude tracking than the GT control.
Beyond these fairly general observations, the authors do not have any further insight and we make no claim of the superiority of either control in practice.
We note furthermore, that these insights are based on a single simple (albeit a Monte-Carlo simulation) example and care should be taken to read too much into the results.

\begin{figure}[!tb]
   \includegraphics[width=1.0\linewidth]{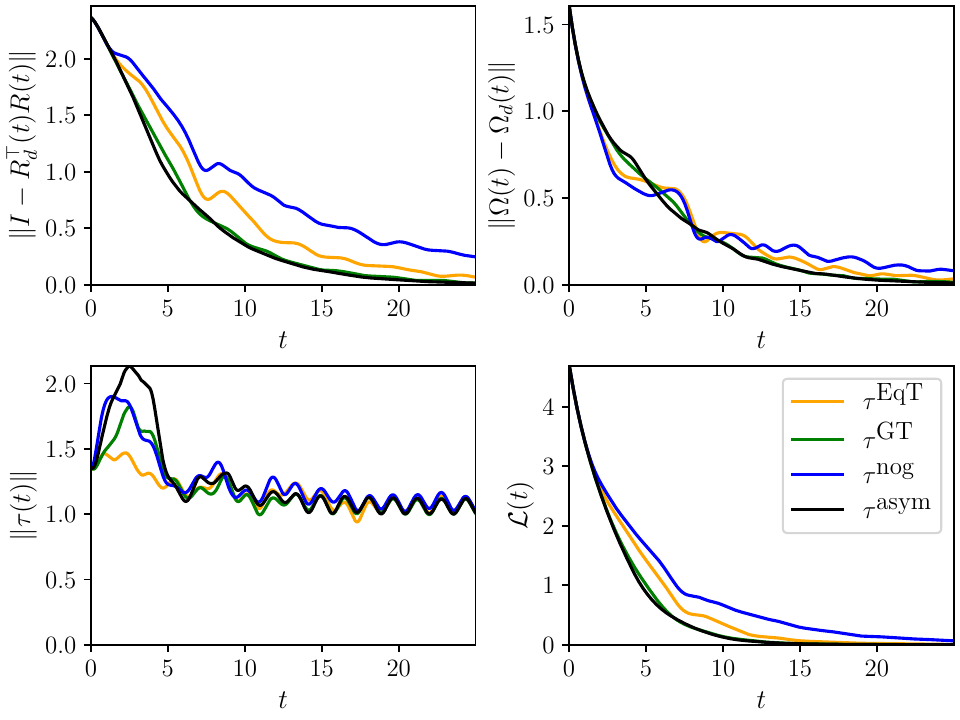}
   \centering
   \caption{$\SO(3)$ tracking error averaged over 200 iterations.}
   \label{fig:so3_error_orthogonal}
\end{figure}

\section{Conclusion}
In this paper, we have addressed the problem of trajectory tracking for general mechanical systems on Lie groups.
To do this, we first addressed the stabilisation problem for a class of systems with time-dependent inertias.
Then, we extended the natural symmetry of the configuration space to a semi-direct product symmetry on the phase space $\grpG \times \gothg^\ast$.
We have shown that the Lie-Poisson dynamics are equivariant with respect to this symmetry, leading to a natural error construction for trajectory tracking.
Equivariance of the dynamics implies that the error system is itself a Lie-Poisson system and thus the tracking problem can be reduced to the stabilisation problem.
To verify the approach, we explored an example system on $\SO(3)$, giving explicit constructions and verifying the control approach in simulation.

\section*{Acknowledgement}
This research was supported by the Australian Research Council through Discovery Grant DP210102607 `` Exploiting the Symmetry of Spatial Awareness for 21st Century Automation''.

\bibliographystyle{plainnat}        
\bibliography{references_archive}           

\appendix

\section{Semidirect structure}
\label{sec:app_semidirect}
We compute some useful identities for the semidirect group $\Gstar$.
Writing $Y \coloneqq (Q_Y, P_Y) \in \Gstar$ and \hbox{$U = (V_U, P_U) \in \gothg \times \gothg^\ast$}, one can also verify the following identities
\begin{align}
    \tD \tL_Y[U] & = \ddt|_{t = 0} (Q_Y, P_Y) (\exp(t V_U), t P_U) \notag                             \\
               & =  (Q_Y V_U, \ad^\ast_{V_U}  P_Y + P_U) \label{eq:DL}                              \\
    \tD R_Y[U] & = \ddt|_{t = 0} ( \exp(t V_U), t P_U) (Q_Y, P_Y) \notag                            \\
               & = (V_U Q_Y, \Ad^\ast_{Q_Y} P_U) \label{eq:DR}                                      \\
    \Ad_Y U    & = \tD \tR_{Y^{-1}}[\tD \tL_Y[U]] \notag                                                \\
               & = (\Ad_{Q_Y} V_U , \Ad^\ast_{Q^{-1}_Y} (\ad^\ast_{V_U}  P_Y + P_U)). \label{eq:Ad}
\end{align}
Note that $\Ad_Y U$ is the adjoint on $\Gstar$ while $\Ad_{Q_Y} V_U$ and  $\Ad^\ast_{Y_Q} Y_P$ is the adjoint on $\grpG$.
Writing $W = (V_W, P_W) \in \gothg \times \gothg^\ast$ then, one also has
\begin{align}
    \ad_W U & = \ddt|_{t=0} \Ad_{(\exp(t V_W), t P_W)} U \notag                         \\
            & = (\ad_{V_W} V_U, \ad^\ast_{V_U} P_W - \ad^\ast_{V_W} P_U). \label{eq:ad}
\end{align}

\end{document}